\documentclass[11pt]{amsart}

\usepackage{euler}

\usepackage{bbm,amsfonts,amsmath,amssymb,mathrsfs,tikz,hyperref,dsfont,
csquotes,enumerate,subfigure} 
\hypersetup{pdfborder={0000}, colorlinks=true, linkcolor=blue,citecolor=citegreen}
\definecolor{citegreen}{rgb}{0.2,0.2,0.6}
\usepackage[babel=true,kerning=true]{microtype}

\hypersetup{
 colorlinks=true,
 citecolor=darkred,
 linkcolor=blue,
 urlcolor=blue}

\usepackage[shortlabels]{enumitem}

\usepackage[font=small]{caption}

\usepackage{palatino}
\usepackage{eucal}
\usepackage{verbatim}

\usepackage{amsmath,amssymb,amsthm,epstopdf,mathrsfs,url}
\usepackage{graphicx}

\definecolor{darkred}{rgb}{0.6,0.0,0.0}
\definecolor{darkred2}{rgb}{0.6,0.0,0.0}

\usepackage{tikz}
\usetikzlibrary{shapes,backgrounds}
\usepackage{mathtools}
\mathtoolsset{showonlyrefs=true}

\newcommand{\sfI}{\mathsf{I}}
\newcommand{\frr}{\mathfrak{r}}

\newcommand{\sfA}{\mathsf{A}}
\newcommand{\sfB}{\mathsf{B}}
\newcommand{\sfP}{\mathsf{P}}

\newcommand\kp{\kappa}
\newcommand\rank{{\rm rank}\,}

\newcommand\arr{\rightarrow}
\newcommand\lbra{\left}
\newcommand\rbra{\right}
\newcommand\bfo{\mathbf{0}}
\def\aa{\alpha}

\def\iff{{\it if, and only if}\,}

\def\pp{\prime}
\def\dpp{{\prime\prime}}

\newcommand\sfJ{\mathsf{J}}

\newcommand\sfY{\mathsf{Y}}
\newcommand\sfT{\mathsf{J}}
\newcommand\Op{\mathsf{H}}
\newcommand\frm{\mathfrak{h}}
\newcommand\p{\partial}
\newcommand\sd{\sigma_{\rm d}}
\newcommand\s{\sigma}
\renewcommand\arr{\rightarrow}
\newcommand\sfe{\mathsf{e}}

\def\s{\sigma}
\def\sd{\sigma_{\rm d}}

\def\lm{\lambda}

\def\softness{0.4}

\definecolor{softred}{rgb}{1,\softness,\softness}
\definecolor{softgreen}{rgb}{\softness,1,\softness}
\definecolor{softblue}{rgb}{\softness,\softness,1}
\definecolor{softrg}{rgb}{1,1,\softness}
\definecolor{softrb}{rgb}{1,\softness,1}
\definecolor{softgb}{rgb}{\softness,1,1}

\newcounter{counter_a}
\newenvironment{myenum}{\begin{list}{{\rm(\roman{counter_a})}}%
{\usecounter{counter_a}
\setlength{\itemsep}{1.ex}\setlength{\topsep}{1.0ex}
\setlength{\leftmargin}{5ex}\setlength{\labelwidth}{5ex}}}{\end{list}}

\newcommand{\eg}{{\it e.g.}\,}

\newcommand{\cf}{{\it cf.}\,}

\numberwithin{figure}{section}
\numberwithin{equation}{section}
\theoremstyle{plain}
\newtheorem*{thm*}{Theorem}
\newtheorem{thm}{Theorem}[section]

\newtheorem{lem}[thm]{Lemma}
\newtheorem{prop}[thm]{Proposition}
\newtheorem{example}[thm]{Example}

\newtheorem{cor}[thm]{Corollary}

\newtheorem{assumption}[thm]{Assumption}
\newtheorem{dfn}[thm]{Definition}
\theoremstyle{remark}
\newtheorem{remark}[thm]{Remark}

\theoremstyle{plain}

\newcommand{\dd}{\mathrm{d}}

\newcommand{\beu}{\begin{equation*}}
\newcommand{\eeu}{\end{equation*}}
\newcommand{\besu}{\begin{equation*}
\begin{aligned}}
\newcommand{\eesu}{\end{aligned}
\end{equation*}}
\newcommand{\bes}{\begin{equation}
\begin{aligned}}
\newcommand{\ees}{\end{aligned}
\end{equation}}

\newcommand\cD{\mathcal D}
\newcommand\cF{\mathcal F}
\newcommand\cG{\mathcal G}
\newcommand\cH{\mathcal H}

\newcommand\cN{\mathcal N}

\newcommand\cR{\mathcal R}
\newcommand\cS{\mathcal S}

\newcommand\frq{\mathfrak q}
\newcommand\frs{\mathfrak s}
\newcommand\frp{\mathfrak p}

\newcommand\eps{\varepsilon}

\newcommand\ov{\overline}
\newcommand\wt{\widetilde}
\newcommand\wh{\widehat}

\newcommand\sess{\sigma_{\rm ess}}

\newcommand\sign{{\rm sign\,}}

\newcommand\void[1]{}

\def\ov{\overline}
\def\eps{\varepsilon}
\def\sess{\sigma_{\rm ess}}

   \def\sE{{\mathfrak E}}

\def\frs{{\mathfrak s}}

      \def\dC{{\mathbb C}}
\def\dD{{\mathbb D}}

   \def\dN{{\mathbb N}}   
      \def\dR{{\mathbb R}}

\def\cD{{\mathcal D}}   \def\cE{{\mathcal E}}   \def\cF{{\mathcal F}}
\def\cG{{\mathcal G}}   \def\cH{{\mathcal H}}   
      
   \def\cN{{\mathcal N}}   \def\cO{{\mathcal O}}
      \def\cR{{\mathcal R}}
\def\cS{{\mathcal S}}      \def\cU{{\mathcal U}}

\newcommand{\dom}{\operatorname{dom}}

\newcommand{\N}{\mathbb{N}}
\newcommand{\R}{{\mathbb{R}}}
\newcommand{\C}{{\mathbb{C}}}

\def\sfT{\mathsf{T}}
\def\sfK{\mathsf{K}}

\def\sfS{\mathsf{S}}
\def\sfR{\mathsf{R}}

\def\bm1{\mathbbm{1}}

\newcommand{\ii}{{\rm i}}

\renewcommand{\ker}{{\rm{ker}\,}}

\renewcommand{\Im}{\text{\rm Im}\,}
\newcommand{\diag}{{\rm diag}}

\newcommand{\comm}[1]{}

\numberwithin{equation}{section}

\title{On resonances and bound states of Smilansky Hamiltonian}

\author{P.~Exner \and V.~Lotoreichik \and M.~Tater}
\address{Nuclear Physics Institute, Czech Academy of Sciences, 25068 \v Re\v z, Czech Republic}
\email{exner@ujf.cas.cz, lotoreichik@ujf.cas.cz, \and
tater@ujf.cas.cz}

\dedicatory{In memory of B.\,S.~Pavlov (1936--2016)}              
	
	
	
\begin{document}

\begin{abstract}
	We consider the self-adjoint Smilansky Hamiltonian $\Op_\eps$ 
	in $L^2(\dR^2)$ 
	associated with the formal differential	expression 
	$-\p_x^2 - \frac12\big(\p_y^2 + y^2) - \sqrt{2}\eps y \delta(x)$
	in the sub-critical
	regime, $\eps \in (0,1)$. We demonstrate the existence 
	of resonances for $\Op_\eps$ 
	on a countable subfamily of sheets of the underlying Riemann surface 
	whose distance from the physical sheet is finite.
	On such sheets, we find resonance free 
	regions and characterise resonances for small $\eps > 0$.
	In addition, we refine the previously known results on the bound states of 
	$\Op_\eps$ in the weak coupling regime ($\eps\arr 0+$).
	In the proofs we use Birman-Schwinger principle for $\Op_\eps$, 
	elements of spectral theory for Jacobi matrices, and the analytic 
	implicit function 
	theorem.
\end{abstract}

\maketitle

\section{Introduction}
In this paper we investigate resonances and bound states of the self-adjoint Hamiltonian $\Op_\eps$ acting 
in the Hilbert space $L^2(\R^2)$ and corresponding to the formal differential expression
\begin{equation}\label{eq:formal}
	-\p_x^2 - \frac12\big(\p_y^2 + y^2) - 
	\sqrt{2}\eps y \delta(x)\qquad \text{on}~\dR^2,
\end{equation}
in the sub-critical regime, $\eps \in (0,1)$. The operator $\Op_\eps$
will be rigorously introduced in Section~\ref{sec:Smilansky} below.
Operators of this type were suggested by U.~Smilansky in~\cite{Sm04}
as a model of \emph{irreversible quantum system}. His aim was to demonstrate that the `heat bath' 
need not have an infinite number of degrees of freedom. 
On a physical level of rigour he showed that the spectrum undergoes 
an abrupt transition at the critical value $\eps = 1$. A mathematically precise 
spectral analysis of these operators and their generalisations has been performed by M.~Solomyak 
and his collaborators 
in~\cite{ES05I, ES05II, NS06, RS07, S04FAA, S06, S06JPA}. Time-dependent
Schr\"odinger equation generated by Smilansky-type Hamiltonian
is considered in~\cite{G11}.

By now many of the spectral properties of $\Op_\eps$ are understood.
On the other hand, little attention has been paid so far to the fact that such a system 
can also exhibit resonances. The main aim of this paper is to initiate investigation of 
these resonances starting from demonstration of their existence. 
One of the key difficulties is that this model belongs to a class wherein the
resolvent extends to a \emph{Riemann 
surface} having uncountably many sheets. The same complication
appears \eg in studying resonances 
for quantum waveguides~\cite{APV00, DEM01, E02, EGST96},~\cite[\S 3.4.2]{EK15} and 
for general manifolds with cylindrical ends~\cite{C02, C04}.

In this paper we prove the existence and obtain a characterisation of resonances
of $\Op_\eps$ on a countable subfamily 
of sheets whose distance from the physical sheet is finite in the sense explained below. On any such sheet we characterise a region which is free of resonances.
As $\eps \arr 0+$, the resonances on such sheets are localized in the vicinities of the thresholds
$\nu_n = n+1/2$, $n\in\dN$. We obtain a description
of the subset of the thresholds in the vicinities of which a resonance
exists for all sufficiently small $\eps > 0$ and derive asymptotic expansions of these resonances in the limit $\eps \arr 0+$. No attempt has been made here
to define and study resonances on the sheets whose distance from
the physical sheet is infinite.

As a byproduct, we obtain refined properties of the bound states
of $\Op_\eps$ using similar methods as for resonances.
More precisely, we obtain
a lower bound on the first eigenvalue of $\Op_\eps$
and an asymptotic expansion of the weakly coupled bound state
of $\Op_\eps$ in the limit $\eps \arr 0+$.

Methods developed in this paper can also be useful to tackle
resonances for the analogue of Smilansky model with regular potential 
which is suggested in~\cite{EB12} and further investigated in~\cite{BE14, BEKT15}.

\subsection*{Notations}
We use notations $\dN := \{1,2,\dots\}$ and
$\dN_0 := \dN\cup\{0\}$ for the sets of positive 
and natural integers, respectively.
We denote the complex plane by $\dC$ and define its commonly used
sub-domains: $\dC^\times := \dC\setminus \{0\}$, $\dC_\pm := \{\lm\in\dC\colon \pm \Im\lm > 0\}$ and
$\dD_r(\lm_0) := \{\lm\in\dC\colon |\lm - \lm_0| < r \}$,
$\dD_r^\times(\lm_0) := \{\lm\in\dC\colon 0 < |\lm - \lm_0| < r \}$, $\dD_r := \dD_r(0)$, $\dD_r^\times := \dD_r^\times(0)$
with $r> 0$. 
The principal value
of the argument for $\lm \in\dC^\times$
is denoted by $\arg\lm\in (-\pi,\pi]$. 
The branches of the square root are defined by
\[
	\dC^\times \ni \lm \mapsto (\lm)^{1/2}_j := 
	|\lm|^{1/2} e^{\ii ((1/2)\arg\lm + j\pi)},\qquad j=0,1.
\] 
If the branch of the square root is not explicitly specified we understand the branch $(\cdot)^{1/2}_0$ by default. We also set $\bfo = (0,0) \in \dC^2$. 

The $L^2$-space over $\dR^d$, $d=1,2$, with the usual inner product
is denoted by $(L^2(\R^d),(\cdot,\cdot)_{\R^d})$
and the $L^2$-based first order Sobolev space by $H^1(\dR^d)$,
respectively.
The space of square-summable sequences of vectors in a Hilbert space 
$\cG$ is denoted by $\ell^2(\N_0;\cG)$. In the case that 
$\cG = \dC$ we simply write $\ell^2(\N_0)$ and denote by $(\cdot,\cdot)$ the usual inner product on it.

%

For $\xi = \{\xi_n\} \in \ell^2(\dN_0)$ we adopt the convention that $\xi_{-1} = 0$.
\emph{Kronecker symbol} is denoted by $\delta_{nm}$, $n,m\in\N_0$,
we set $\sfe_n := \{\delta_{nm}\}_{m\in\N_0}\in \ell^2(\N_0)$, $n\in\N_0$,
and adopt the convention that $\sfe_{-1} := \{0\}$. 
We understand by $\diag (\{q_n\})$ the
\emph{diagonal matrix} in $\ell^2(\dN_0)$ 
with entries $\{q_n\}_{n\in\N_0}$ 
and by $\sfJ(\{a_n\},\{b_n\})$ the \emph{Jacobi matrix} in $\ell^2(\N_0)$ with diagonal entries 
$\{a_n\}_{n\in\N_0}$ and off-diagonal entries $\{b_n\}_{n\in\N}$\footnote{We do not distinguish between Jacobi matrices and operators
in the Hilbert space $\ell^2(\N_0)$ induced by them, since in our considerations all the Jacobi matrices are bounded,  closed, 
and everywhere defined in $\ell^2(\N_0)$.}. 
We also set $\sfJ_0 := \sfJ(\{0\},\{1/2\})$. 

By $\s(\sfK)$ we denote the spectrum of a closed (not necessarily self-adjoint) operator $\sfK$ 
in a Hilbert space. An isolated eigenvalue $\lm\in\dC$ of $\sfK$ 
having finite algebraic multiplicity 
is a point of the \emph{discrete spectrum} for $\sfK$; see~\cite[\S XII.2]{RS-IV} for details.
The set of all the points of the discrete spectrum for $\sfK$
is denoted by $\sd(\sfK)$ and the \emph{essential spectrum} of $\sfK$ is defined by $\sess(\sfK) := \s(\sfK)\setminus\sd(\sfK)$.
 
For a self-adjoint operator $\sfT$ in a Hilbert space
we set $\lm_{\rm ess}(\sfT) := \inf\sess(\sfT)$ and, for $k\in\dN$, $\lm_k(\sfT)$ 
denotes the $k$-th eigenvalue of $\sfT$
in the interval $(-\infty, \lm_{\rm ess}(\sfT))$. These eigenvalues are 
ordered non-decreasingly with multiplicities taken into account.
The number of eigenvalues with multiplicities
of the operator $\sfT$ lying in a closed, open, or half-open interval
$\Delta\subset\R$ satisfying $\sess(\sfT) \cap \Delta = \varnothing$
is denoted by $\cN(\Delta;\sfT)$. 
For $\lm \le \lm_{\rm ess}(\sfT)$ the \emph{counting function} of $\sfT$ is defined by
$\cN_\lm(\sfT) := \cN((-\infty,\lm);\sfT)$.


\subsection{Smilansky Hamiltonian}\label{sec:Smilansky}
Define the \emph{Hermite functions} 
\begin{equation}\label{eq:Hermite}
	\chi_n(y) := e^{-y^2/2}H_n(y),\qquad n\in\N_0.
\end{equation}
Here, $H_n(y)$ is the \emph{Hermite polynomial} 
of degree $n\in\N_0$ normalized by the condition $\|\chi_n\|_\R = 1$\footnote{This normalization means that $H_n(y)$ is, in fact, a product
of what is usually called the \emph{Hermite polynomial} of degree $n\in\N_0$
with a normalization constant which depends on $n$.}.
For more details on Hermite polynomials 
see~\cite[Chap. 22]{AS} and also~\cite[Chap. 5]{Ch78}.
As it is well-known, the family $\{\chi_n\}_{n\in\N_0}$ 
constitutes an orthonormal basis
of $L^2(\R)$. Note also that the functions $\chi_n$ satisfy the three-term recurrence relation
\begin{equation}\label{eq:reccurence}
	\sqrt{n+1}\chi_{n+1}(y) - \sqrt{2}y\chi_{n}(y) + \sqrt{n}\chi_{n-1}(y) = 0,
	\qquad n\in\dN_0,
\end{equation}
where we adopt the convention $\chi_{-1} \equiv 0$.
The relation~\eqref{eq:reccurence} can be easily deduced from the recurrence relation~\cite[eq. 22.7.13]{AS} for Hermite polynomials.
By a standard argument any function $U\in L^2(\R^2)$ admits unique expansion
\begin{equation}\label{eq:Uexpansion}
	U(x,y) = \sum_{n\in\N_0} u_n(x)\chi_n(y),
	\qquad 
	u_n(x) := \int_\R U(x,y)\chi_n(y)\dd y,
\end{equation}
where $\{u_n\} \in \ell^2(\dN_0; L^2(\R))$. Following the presentation
in~\cite{S06}, we identify the function $U\in L^2(\R^2)$ and the sequence $\{u_n\}$
and write $U\sim \{u_n\}$. This identification defines a natural unitary transform
between the Hilbert spaces $L^2(\R^2)$ and $\cH := \ell^2(\dN_0;L^2(\R))$.
For the sake of brevity,
we denote the inner product on $\cH$ by $\langle\cdot, \cdot\rangle$.
Note that the Hilbert space $\cH$ can also be viewed as the tensor product
$\ell^2(\dN_0)\otimes L^2(\dR)$.

For any $\eps \in\dR$ we define the subspace $\cD_\eps$ of 
$\cH$ as follows: an element $U \sim \{u_n\}\in\cH$ belongs to $\cD_\eps$
\iff
\begin{myenum}
	\item $u_n\in H^1(\R)$ for all $n\in\dN_0$; 
	\item $\{- (u_{n,+}^\dpp\oplus u_{n,-}^\dpp) + \nu_n u_n\} \in\cH$
	with $u_{n,\pm} := u_n|_{\R_\pm}$ and $\nu_n = n+1/2$ for $n\in\dN_0$;
	\item the boundary conditions
	\[
		u_n^\pp(0+) - u_n^\pp(0-) =
		\eps\big(\sqrt{n+1} u_{n+1}(0) + \sqrt{n}u_{n-1}(0)\big)
	\]
	are satisfied for all $n\in\dN_0$. For $n= 0$
	only the first term is present on the right-hand side.
\end{myenum}
By~\cite[Thm. 2.1]{S06}, the operator 
\begin{equation}\label{eq:operator}
	\dom\Op_\eps := \cD_\eps,\qquad \Op_\eps \{u_n\} := \{- (u_{n,+}^\dpp\oplus u_{n,-}^\dpp) + \nu_n u_n\},
\end{equation}
is self-adjoint in $\cH$. It corresponds to the formal 
differential expression~\eqref{eq:formal}. 
Further, we provide another way of defining $\Op_\eps$ which makes the correspondence
between the operator $\Op_\eps$ 
and the formal differential expression~\eqref{eq:formal} more transparent. 
To this aim we define the straight line $\Sigma := \{(0,y)\in\dR^2\colon y\in\dR\}$.
Then the Hamiltonain $\Op_\eps$, $\eps \in (-1,1)$, can be alternatively
introduced as the unique self-adjoint operator in $L^2(\R^2)$  
associated via the first representation
theorem~\cite[Thm. VI.2.1]{Kato} with a closed, densely defined, symmetric, and 	
semi-bounded quadratic form
\begin{equation}\label{eq:frm}
\begin{split}
	\frm_\eps[u] 
	& := 
	\|\p_x u\|^2_{\R^2} 
	+ 
	\frac12\|\p_y u\|^2_{\R^2} + \frac12 (y u, y u)_{\R^2} 
	+ 
	\eps\sqrt{2}\lbra( \sign(y) |y|^{1/2} u|_\Sigma, |y|^{1/2} u|_\Sigma \rbra)_\R,
	\\
	\dom\frm_\eps 
	& :=
	\lbra\{u\in H^1(\R^2)\colon yu \in L^2(\R^2), |y|^{1/2}(u|_\Sigma) 
	\in L^2(\R) \rbra\}.
\end{split}	
\end{equation}
For more details
and for the proof of equivalence between the two definitions of $\Op_\eps$
see~\cite[\S 9]{S06}. 
Since $\Op_\eps$ commutes with the parity operator in $y$-variable, it 
is unitarily equivalent to $\Op_{-\eps}$. We remark
that the case $\eps = 0$ admits separation of variables.
Thus, it suffices to study $\Op_\eps$ with $\eps > 0$.

In the following proposition we collect fundamental
spectral properties of $\Op_\eps$, $\eps \in (0,1)$, which are of importance
in the present paper.
\begin{prop}\label{prop:known}
	Let the self-adjoint operator $\Op_\eps$, $\eps\in(0,1)$,
	be as in~\eqref{eq:operator}.
	Then the following claims hold:
	\begin{myenum}
		\item $\sess(\Op_\eps) = [1/2,+\infty)$;
		\item $\inf\s(\Op_\eps) \ge \frac{1 - \eps}{2}$;
		\item $1 \le \cN_{1/2}(\Op_\eps) < \infty$;
		\item  $\cN_{1/2}(\Op_\eps) = 1$ for all sufficiently small $\eps > 0$.
	\end{myenum}		
\end{prop}
Items~(i)-(iii) follow from~\cite[Lem 2.1]{S04FAA}
and~\cite[Thm. 3.1\,(1),(2)]{S06}. Item~(iv) is a consequence of~\cite[Thm. 3.2]{S04FAA} and~\cite[\S 10.1]{S06}.
Although we only deal with the sub-critical case, $\eps \in(0,1)$, 
we remark that in the critical case, $\eps = 1$, the spectrum of $\Op_1$ equals to $[0,+\infty)$
and that in the sup-critical case, $\eps > 1$, the spectrum of $\Op_\eps$ covers the whole real axis.
Finally, we mention that in most of the existing literature
on the subject not $\eps > 0$ itself but $\aa = \sqrt{2}\eps$ is 
chosen as the coupling parameter. We choose another normalization of the coupling
parameter in order to simplify formulae in the proofs of the main results.

\subsection{Main results}
While we are primarily interested in the resonances, as indicated in the introduction, we have also a claim to make about 
the discrete spectrum which we present here as our first main result and which complements the results listed 
in Proposition~\ref{prop:known}. 
\begin{thm}\label{thm:discr}
	Let the self-adjoint operator $\Op_\eps$, $\eps\in(0,1)$, 
	be as in~\eqref{eq:operator}. Then the following claims hold.
	\begin{myenum}
		\item $\lm_1(\Op_\eps) \ge  
		1 - \sqrt{\frac{1}{4}+ \eps^4}$ 
		for all $\eps \in(0,1)$.
		\item $	\lm_1(\Op_\eps) 
		= \nu_0 - \frac{\eps^4}{16} + \cO(\eps^5)$
		as $\eps\arr 0+$.
	\end{myenum}
\end{thm}
Theorem~\ref{thm:discr}\,(i) is proven by means of 
Birman-Schwinger principle. The bound in Theorem~\ref{thm:discr}\,(i) 
is non-trivial for $\eps^4 < 3/4$. This bound
is better than the one in Proposition~\ref{prop:known}\,(ii) for small $\eps > 0$.

For the proof of Theorem~\ref{thm:discr}\,(ii) we combine 
Birman-Schwinger principle and the analytic implicit function theorem.
We expect that the error term $\cO(\eps^5)$ in Theorem~\ref{thm:discr}\,(ii)
can be replaced by  $\cO(\eps^6)$ because the operator $\Op_\eps$ 
has the same spectral properties as $\Op_{-\eps}$ for any $\eps\in (0,1)$. 
Therefore, the expansion of $\lm_1(\Op_{\eps})$ 
must be invariant with respect to interchange between $\eps$ and $-\eps$. 
In Lemma~\ref{lem:aux1} given in Section~\ref{sec:weak} we derive an
implicit scalar equation on $\lm_1(\Op_\eps)$. This equation 
gives analyticity of $\eps\mapsto\lm_1(\Op_\eps)$ for small $\eps$.
It can also be used to compute higher order
terms in the expansion of $\lm_1(\Op_\eps)$.
However, these computations might be quite tedious.

Our second main result concerns the resonances of $\Op_\eps$. 
Before formulating it, we need to define the resonances rigorously.
Let us consider the sequence of functions 
\begin{equation}\label{eq:r_n}
	r_n(\lm) := (\nu_n -\lm)^{1/2}, \qquad n\in\dN_0.
\end{equation}	
Each of them has two branches $r_n(\lm,l) := (\nu_n -\lm)^{1/2}_l$, $l=0,1$.
The vector-valued function $R(\lm) = (r_0(\lm),r_1(\lm),r_2(\lm),\dots)$
naturally defines the Riemann surface $\wh Z$ with uncountably many sheets.
With each sheet of $\wh Z$ we associate the set $E \subset \dN_0$
and the \emph{characteristic vector} $l^E$ defined as
\begin{equation}
	l^E := \{l_0^E,l_1^E,l_2^E,\dots\},\qquad 
	l_n^E := 
	\begin{cases}
	 	0,\qquad  n\notin E,\\
		1,\qquad  n\in E.\\
	 \end{cases}
\end{equation}
We adopt the convention that $l_{-1}^E = 0$.
The respective sheet of $\wh Z$ is convenient to denote by $Z_E$.
Each sheet $Z_E$ of $\wh Z$ can be identified with the set $\dC\setminus [\nu_0,+\infty)$
and we denote by $Z_E^\pm$ the parts of $Z_E$ corresponding to $\dC_\pm$. 
With the notation settled we define the realization of $R(\cdot)$
on $Z_E$ as
\begin{equation}\label{eq:RE}
	R_E(\lm) := 
	(r_0(\lm,l_0^E),r_1(\lm,l_1^E),r_2(\lm,l_2^E),\dots).
\end{equation}
The sheets $Z_E$ and $Z_F$ are \emph{adjacent} 
through the interval $(\nu_{n}, \nu_{n+1})\subset\dR$, $n\in\dN_0$,
($Z_E\sim_n Z_F$), if their characteristic vectors $l^E$ and $l^F$ satisfy
\[
\begin{split}
	l^F_k = 1 - l^E_k,&\qquad \text{for}~ k=0,1,2,\dots,n\\
	l^F_k = l^E_k,&\qquad \text{for}~ k > n.
\end{split}	
\]
We set $\nu_{-1} = -\infty$ and note that any sheet
$Z_E$ is adjacent to itself through $(\nu_{-1}, \nu_0)$.
In particular, the function $\lm\mapsto R_E(\lm)$ 
turns out to be componentwise analytic on the Riemann surface $\wh Z$.

The sequence $\sE = \{E_1,E_2,\dots, E_N\}$ of subsets of $\dN_0$
is called a \emph{path} if for any $k = 1,2,\dots, N-1$
the sheets $Z_{E_k}$ and $Z_{E_{k+1}}$  are adjacent. The following
discrete metric   
\begin{equation}\label{eq:rho}
	\rho(E,F) := 
	\inf\{N\in\dN_0 \colon \sE = \{E_1,E_2,\dots, E_N\},
	E_1 = E, E_N = F\}, 
\end{equation}
turns out to be convenient. 
The value $\rho(E,F)$ equals the number
of sheets in the shortest path connecting $Z_E$ and $Z_F$. Note that
for some sheets $Z_E$ and $Z_F$ a path between them does not exist 
and in this case we have
$\rho(E,F) = \infty$.
We identify the \emph{physical sheet} with the sheet $Z_\varnothing$
(for $E = \varnothing$).
A sheet $Z_E$ of $\wh Z$ is \emph{adjacent} to the physical sheet $Z_\varnothing$ if $\rho(E,\varnothing) = 1$
and it can be characterised by existence of $N\in\dN_0$ such that $l^E_n = 1$ \iff $n\le N$.
Further, we define the component 
\begin{equation}\label{eq:wtZ}
	\wt Z := \cup_{E\in \cE} Z_E 
	\subset \wh Z, \qquad \cE := \{E\subset\dN_0\colon 	
	\rho(E,\varnothing) < \infty\},
\end{equation}	 
of $\wh Z$ which plays a distinguished role in our considerations.
Any sheet in $\wt Z$ is located on a finite distance from the physical
sheet with respect to the metric $\rho(\cdot,\cdot)$.
The component $\wt Z$ of $\wh Z$ in~\eqref{eq:wtZ} can be
alternatively characterised as
\begin{equation}
	\wt Z = \cup_{F \in \cF} Z_F,
	\qquad \cF := \{F\subset\dN_0\colon \sup \{n\in\dN_0\colon l_n^F = 1\}  < \infty\}.
\end{equation}
%
The number of the sheets in $\wt Z$ is easily seen to be countable.
In order to define the resonances of $\Op_\eps$ on $\wt Z$ we show that the resolvent of $\Op_\eps$ admits an extension to $\wt Z$ in a certain weak sense. 
\begin{prop}\label{prop:cont}
	For any $u\in L^2(\dR)$ and $n\in\dN_0$ the function
	\begin{equation}\label{eq:frr}
		\lm \mapsto \frr_{n,\eps}^\varnothing(\lm;u) := 
		\big\langle(\Op_\eps - \lm)^{-1} u\otimes \sfe_n, u\otimes \sfe_n\big\rangle
	\end{equation}
	admits unique meromorphic continuation $\frr_{n,\eps}^E(\cdot;u)$
	from the physical sheet $Z_\varnothing$ to any sheet $Z_E \subset \wt Z$.
\end{prop}
The proof of Proposition~\ref{prop:cont} is postponed until Appendix~\ref{app:A}. 
Now we have all the tools to define resonances of $\Op_\eps$ on $\wt Z$.
\begin{dfn}\label{def:resonances}
	Each resonance of $\Op_\eps$ on $Z_E\subset \wt Z$
	is identified with a pole of $\frr_{n,\eps}^E(\cdot;u)$ for 
	some $u\in L^2(\dR)$ and $n\in\dN_0$.
	The set of all the resonances for $\Op_\eps$ on the sheet $Z_E$
	is denoted by $\cR_E(\eps)$.
\end{dfn}
Our definition of resonances for $\Op_\eps$ is consistent with~\cite[\S XII.6]{RS-IV}, 
see also~\cite[Chap. 2]{EK15} and~\cite{M93} for multi-threshold case. 
It should be emphasized that by the spectral theorem for self-adjoint operators
eigenvalues of $\Op_\eps$ are also regarded as resonances in the sense of Definition~\ref{def:resonances} lying on the physical sheet $Z_\varnothing$. 
This allows us to treat the eigenvalues and `true' resonances on the same
footing. Needless to say, bound states and true resonances correspond to different
physical phenomena and their equivalence in this paper is merely a useful mathematical
abstraction.

According to Remark~\ref{rem:symmetric} below the set of the resonances 
for $\Op_\eps$ on $Z_E$ is symmetric with respect to the real axis. Thus, it suffices
to analyse resonances on $Z_E^-$. Now we are prepared to formulate the main result on resonances.
%
\begin{thm}\label{thm:res}
	Let the self-adjoint operator $\Op_\eps$, $\eps\in(0,1)$,
	be as in~\eqref{eq:operator}.
	Let the sheet $Z_E\subset \wt Z$
	of the Riemann surface $\wh Z$ be fixed.
	Define 	the associated set by
	\[
		\cS(E)  := \big\{n\in\dN\colon 
		(l_{n-1}^E, l_n^E,l_{n+1}^E) \in \{(1,0,0),(0,1,1)\}\big\}.
	\]
	Let $\cR_E(\eps)$ be as in Definition~\ref{def:resonances}
	and set $\cR_E^-(\eps) := \cR_E(\eps) \cap\dC_-$.
	Then the following claims hold.
	\begin{myenum}
		\item $\cR_E^-(\eps)
		\subset \cU(\eps) :=
		\lbra\{\lm\in \dC_-
		\colon 	|\nu_{n-1} -\lm||\nu_n - \lm| \le \eps^4 n^2,~
		\forall n\in\dN\rbra\}$.
		\item For any $n\in\cS(E)$ and sufficiently small $\eps > 0$
		there is exactly one resonance $\lm_n^E(\Op_\eps) \in \dC_-$
		of $\Op_\eps$ on $Z_E^-$ lying in a neighbourhood of $\nu_n$, 
		with the expansion
		\begin{equation}\label{eq:res_asymp}
			\lm_n^E(\Op_\eps) 
			= 
			\nu_n - \frac{\eps^4}{16}\big[(2n + 1)+  2n(n+1)\ii\big] +  \cO(\eps^5),
			\qquad\eps\arr 0+.
		\end{equation} 
		\item For any $n\in\dN\setminus \cS(E)$ 
		and all sufficiently small $\eps,r > 0$
		\[
			\cR_E^-(\eps)\cap \dD_r(\nu_n) = \varnothing.
		\]
	\end{myenum}
\end{thm}
\begin{figure}
\centering
\includegraphics[width=0.7\textwidth]{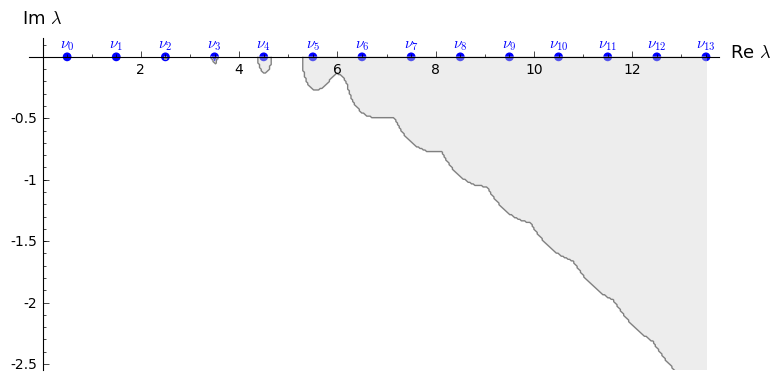}
\caption{The region $\cU(0.12)$ (for $\eps = 0.12$) from Theorem~\ref{thm:res}\,(i) (in grey) consists of $6$ connected components. 
The components located in the neighbourhoods of the points $\nu_0,\nu_1$, $\nu_2$, $\nu_3$,
are not visible because of being too small. The plot is performed with the aid of \emph{Sagemath}.}
\label{fig:free}
\end{figure}
In view of Theorem~\ref{thm:res}\,(i) for  sufficiently small $\eps > 0$
the resonances of $\Op_\eps$ on any sheet of $\wt Z$
are located in some neighbourhoods of the thresholds $\nu_n$
(see Figure~\ref{fig:free}).
Such behaviour is typical for problems with many thresholds; see \eg \cite{DEM01, EGST96} and~\cite[\S 2.4, 3.4.2]{EK15}. 
Note also that the estimate in Theorem~\ref{thm:res}\,(i)
reflects the correct order in $\eps$ in the weak coupling limit $\eps \arr 0+$
given in Theorem~\ref{thm:res}\,(ii).
However, the coefficient of $\eps^4$ in the definition
of $\cU(\eps)$ can be probably improved.
Observe also that $\cR_E^-(\eps) \subset \cU(1)$
for any $\eps \in(0,1)$.

According to Theorem~\ref{thm:res}\,(ii)-(iii) existence
of a resonance near the threshold $\nu_n$,  $n\in\dN$, on a sheet $Z_E$ for small $\eps > 0$ depends only on the branches chosen for $r_{n-1}(\lm)$, $r_n(\lm)$, $r_{n+1}(\lm)$
on $Z_E$. Although, one can not exclude that higher order terms 
in the asymptotic expansion~\eqref{eq:res_asymp}
depend on the branches chosen for other square roots.
By exactly the same reason as in Theorem~\ref{thm:discr}\,(ii), 
we expect that the error term $\cO(\eps^5)$ in Theorem~\ref{thm:res}\,(ii)
can be replaced by $\cO(\eps^6)$.
Theorem~\ref{thm:res}\,(ii)-(iii) are proven by means of the
Birman-Schwinger principle and the analytic implicit function theorem.
The implicit scalar equation on resonances derived in Lemma~\ref{lem:aux1}
gives analyticity of $\eps\mapsto\lm_n^E(\Op_\eps)$ for small $\eps > 0$
and, as in the bound state case, it can be used to compute further terms in the expansion of $\lm_n^E(\Op_\eps)$.

We point out that according
to numerical tests in~\cite{ELT} some resonances emerge from the inner points of the intervals $(\nu_n,\nu_{n+1})$, $n\in\dN_0$, as $\eps\arr 1-$. The mechanism 
of creation for these resonances is unclear at the moment. 
%
\begin{example}
	Let $E = \{1,2,4,5\}$. In this case $l^E = \{0,1,1,0,1,1,0,0,0,0,0,\dots\}$ 
	and we get that $\cS(E) = \{1,4,6\}$. 
	By Theorem~\ref{thm:res}\,(ii)-(iii) for all sufficiently small
	$\eps > 0$ there will be exactly one resonance on $Z_E^-$
	near $\nu_1$, $\nu_4$, $\nu_6$ and no resonances near
	the thresholds $\nu_n$ with $n \in\dN\setminus\{1,4,6\}$. 
	We confirm this result by numerical tests whose outcome is shown 
	in Figures~\ref{fig:weak} and~\ref{fig:weak2}. 

\begin{figure}[htb!]
	\includegraphics[width=0.97 \textwidth]{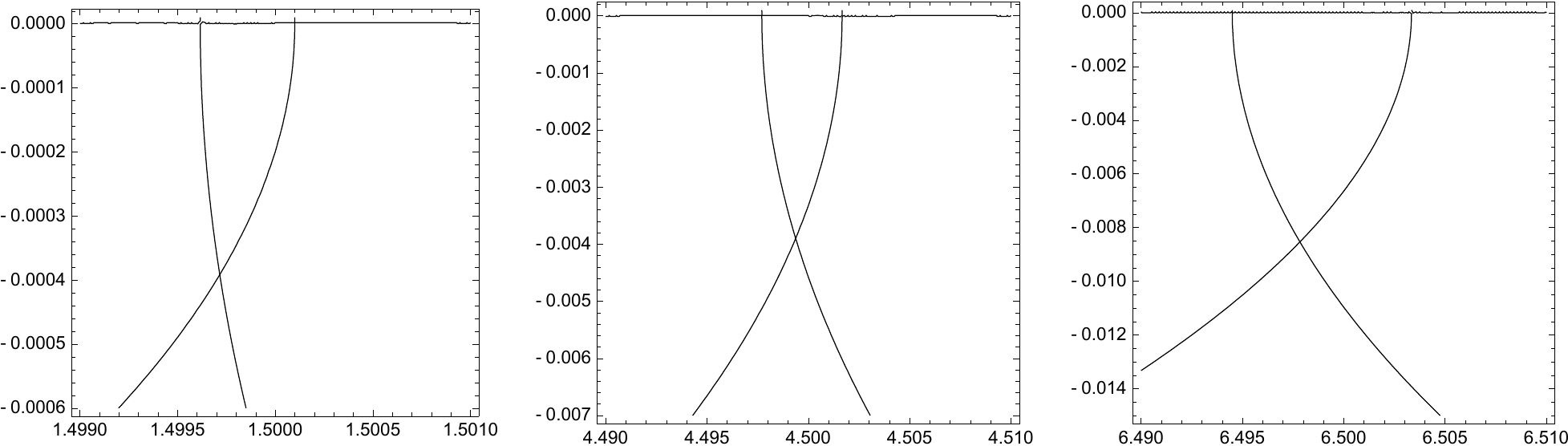}
	\vspace{2ex}
	\caption{Resonances of $\Op_\eps$ with $\eps = 0.2$ lying on $Z_E^-$ with $E = \{1,2,4,5\}$ are computed numerically with the help
	of \emph{Mathematica}. Unique weakly coupled resonances near the thresholds $\nu_1 = 1.5$, $\nu_4 = 4.5$, $\nu_6 = 6.5 $ 
	are located at the intersections of the curves.} \label{fig:weak}
\end{figure}

	To plot Figure~\ref{fig:weak} 
	we used the condition on resonances in Theorem~\ref{thm:BS_resonance} below.
	The infinite Jacobi matrix in this condition was truncated up to a reasonable finite size. 
	Along the curves, respectively, the real and the imaginary part
	of the determinant of the truncated matrix vanishes. At the points of intersection of the curves
	the determinant itself vanishes. These points are expected to be
	close to true resonances\footnote{The analysis of convergence of the numerical method is beyond
	our scope.}. We have also verified numerically that resonances do not 
	exist near other low-lying thresholds $\nu_n$ with $n\in\dN\setminus\{1,4,6\}$
	which corresponds well to Theorem~\ref{thm:res}. In Figure~\ref{fig:weak2}
	we summarise the results of all the numerical tests.

\begin{figure}[htb!]
	\includegraphics[width=0.75 \textwidth]{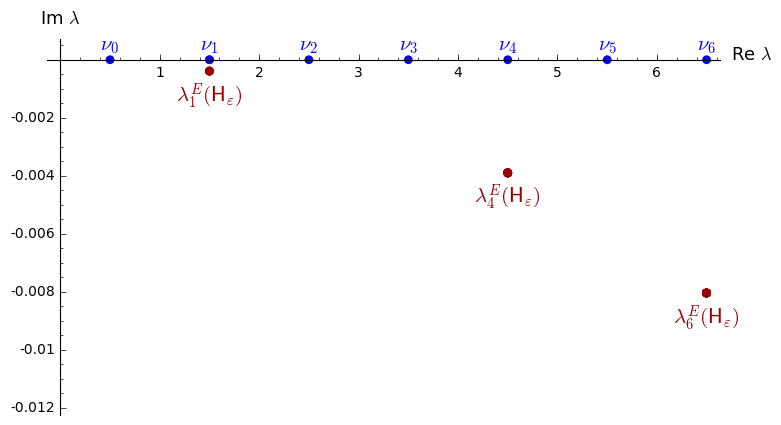}
	\caption{Resonances of $\Op_\eps$ with $\eps = 0.2$ lying on $Z_E^-$ with $E = \{1,2,4,5\}$.} 
	\label{fig:weak2}
\end{figure}

\end{example}
Finally, we mention that no attempt has been made here to
analyse the multiplicities of the resonances
and to investigate
resonances lying on $\wh Z\setminus \wt Z$. 

\subsection*{Structure of the paper}
Birman-Schwinger-type principles 
for characterisation of eigenvalues and resonances of $\Op_\eps$
are provided in Section~\ref{sec:BS}.
Theorem~\ref{thm:discr}\,(i) on a lower bound for the first eigenvalue 
and Theorem~\ref{thm:res}\,(i) on resonance free region are proven 
in Section~\ref{sec:free}. The aim of Section~\ref{sec:weak} is to prove Theorem~\ref{thm:discr}\,(ii) and Theorem~\ref{thm:res}\,(ii)-(iii) on weakly coupled bound states and resonances. 
The proofs of technical statements
formulated in Proposition~\ref{prop:cont} and Theorem~\ref{thm:BS_resonance} are postponed until Appendix~\ref{app:A}.

\section{Birman-Schwinger-type conditions}
\label{sec:BS}

Birman-Schwinger principle is a powerful tool for analysis of the discrete spectrum
of a perturbed operator in the spectral gaps of the unperturbed one. This principle
has also various other applications. Frequently, it can be generalized 
to detect resonances, defined as the poles of a meromorphic continuation of the
(sandwiched) resolvent from the physical sheet to non-physical sheet(s) 
of the underlying Riemann surface. In the model under consideration we encounter yet another manifestation of this principle.

In order to formulate a Birman-Schwinger-type condition on the bound states
for $\Op_\eps$ we introduce the sequence of functions
\begin{equation}\label{eq:bn}
	b_n(\lm) := \frac{n^{1/2}}{2(\nu_n - \lm)^{1/4}(\nu_{n-1} - \lm)^{1/4}}, \qquad n\in\N,
\end{equation}
and the off-diagonal Jacobi matrix
\begin{equation}\label{eq:Jacobi}
	\sfJ(\lm) = \sfJ\lbra(\{0\},\{b_n(\lm)\}\rbra), 
	\qquad \lm \in \lbra(0,\nu_0\rbra).
\end{equation}
Recall that we use the same symbol $\sfJ(\lm)$ for the operator 
in $\ell^2(\N_0)$ generated by this matrix. 
It is straightforward to check that the operator 
$\sfJ(\lm)$ is bounded and self-adjoint.
It can be easily verified that the difference $\sfJ(\lm) - \sfJ_0$ is a compact
operator. Therefore, one has $\sess(\sfJ(\lm)) = \sess(\sfJ_0) = [-1,1]$. 
Moreover, the operator $\sfJ(\lm)$ has simple eigenvalues $\pm \mu_n$, 
$\mu_n > 1$, with the only possible accumulation points at $\mu = \pm 1$.
%
\begin{thm}\cite[Thm. 3.1]{S04FAA}\label{thm:BS}
	Let the self-adjoint operator $\Op_\eps$, $\eps \in (0,1)$, 
	be as in~\eqref{eq:operator}
	and let the Jacobi matrix $\sfJ(\lm)$ be as in~\eqref{eq:Jacobi}. 
	Then the relation
	\begin{equation}\label{eq:BS}
		\cN((0,\lm); \Op_\eps) = \cN( (1/\eps,+\infty); \sfJ(\lm) )	
	\end{equation} 
	holds for all $\lm \in (0,\nu_0)$.
\end{thm}
\begin{remark}\label{rem:BS_modified}
	A careful inspection of the proof of~\cite[Thm 3.1]{S04FAA} 
	yields that Theorem~\ref{thm:BS}
	can also be modified, replacing~\eqref{eq:BS} by
	\begin{equation}\label{eq:BS_modified}
		\cN((0,\lm]; \Op_\eps) 
		=
		\cN( [1/\eps,+\infty);\sfJ(\lm) ).	
	\end{equation} 
	In other words, the right endpoint of the interval $(0,\lm)$
	and the left endpoint of the interval $(1/\eps,+\infty)$
	can be simultaneously included.
\end{remark}
The following consequence of Theorem~\ref{thm:BS} and of the above remark will
be useful further.
\begin{cor}\label{cor:BS}
	Let the assumptions be as in Theorem~\ref{thm:BS}. Then the following claims hold:
	\begin{myenum}
		\item $\eps\mapsto\lm_k(\Op_\eps)$ are continuous non-increasing functions;
		\item 
		$\dim\ker(\Op_\eps - \lm) 
		= 
		\dim\ker\big(\sfI + \eps\sfJ(\lm)\big)$ for all $\lm\in (0,\nu_0)$.
		In particular, since the eigenvalues of $\sfJ(\lm)$ are simple, the eigenvalues
		of $\Op_\eps$ are simple as well.
	\end{myenum}	
\end{cor}
\begin{proof}
	{\rm (i)} 
	Let $\eps_1\in (0,1)$.
	For $\lm = \lm_k(\Op_{\eps_1})$, $k\in\N$, we have by Theorem~\ref{thm:BS} 
	and Remark~\ref{rem:BS_modified}
	\[
		\cN([1/\eps_1,+\infty);\sfJ(\lm)) = \cN((0,\lm];\Op_{\eps_1}) \ge k.
	\]
	Hence, for any $\eps_2 \in (\eps_1,1)$ we obtain 
	\[
		\cN((0,\lm];\Op_{\eps_2})  
		= 
		\cN([1/\eps_2,+\infty);\sfJ(\lm)) 
		\ge 
		\cN([1/\eps_1,+\infty);\sfJ(\lm)) 
		\ge k.
	\]
	Therefore, we get $\lm_k(\Op_{\eps_2}) \le \lm = \lm_k(\Op_{\eps_1})$.
	
	Recall that $\Op_\eps$ represents the quadratic form $\frm_\eps$
	defined in~\eqref{eq:frm}. Continuity of the eigenvalues follows 
	from~\cite[Thms. VI.3.6, VIII.1.14]{Kato}
	and from the fact that the quadratic form 
	\[
		\dom \frm_\eps\ni u \mapsto 
		\eps\sqrt{2}\lbra(\sign y|y|^{1/2}u|_\Sigma, |y|^{1/2}u|_\Sigma\rbra)_\R,
		\qquad \eps\in (0,1),
	\]
	is relatively bounded with respect to
	\[
		\dom \frm_\eps\ni u \mapsto 
		\|\p_x u\|^2_{\R^2} 
		+ 
		\frac12\|\p_y u\|^2_{\R^2} + \frac12 (y u, y u)_{\R^2}
	\]
	with a bound less than one; \cf \cite[Lem. 2.1]{S04FAA}.

	\vspace{1.0ex}
	
	\noindent (ii)
	By Theorem~\ref{thm:BS}, Remark~\ref{rem:BS_modified},
	and using symmetry of $\s(\sfJ(\lm))$ with respect to the origin we get
	\[
	\begin{split}
		\dim\ker(\Op_\eps - \lm) 
		& = \cN((0,\lm]; \Op_\eps) - \cN((0,\lm); \Op_\eps)\\
		& =
		\cN([1/\eps,+\infty); \sfJ(\lm)) - \cN((1/\eps,+\infty); \sfJ(\lm)) 
		=				
		\dim\ker\big(\sfI + \eps\sfJ(\lm)\big). \qedhere
	\end{split}	
	\]         
\end{proof}
For resonances of $\Op_\eps$ one can derive a Birman-Schwinger-type condition  
analogous to the one in Corollary~\ref{cor:BS}\,(ii).
For the sheet $Z_E\subset \wt Z$ of the Riemann surface $\wh Z$ 
we define the Jacobi matrix
\begin{equation}\label{eq:Jacobi2}
	\sfJ_E(\lm) := \sfJ(\{0\},\{b_{n}^E(\lm)\}),\qquad \lm \in \dC\setminus [\nu_0,+\infty),
\end{equation}
where 
\begin{equation}\label{eq:bnE}
	b_n^E(\lm) := 
	\frac{1}{2}
	\bigg(\frac{n}{r_n(\lm,l_n^E) r_{n-1}(\lm, l_{n-1}^E)}\bigg)^{1/2},
	\qquad n\in\dN.
\end{equation}
The Jacobi matrix $\sfJ_E(\lm)$ in~\eqref{eq:Jacobi2}
is closed, bounded, and everywhere defined in 
$\ell^2(\N_0)$, but in general non-selfadjoint. 
For $E = \varnothing$ and $\lm \in (0,\nu_0)$ 
the Jacobi matrix $\sfJ_\varnothing(\lm)$ coincides with $\sfJ(\lm)$ 
in~\eqref{eq:Jacobi}. In what follows
it is also convenient to set $b_0^E(\lm) = 0$.  
In the next theorem we characterise resonances of $\Op_\eps$ 
lying on the sheet $Z_E$.
\begin{thm}\label{thm:BS_resonance}
	Let the self-adjoint operator $\Op_\eps$, $\eps \in (0,1)$, 
	be as in~\eqref{eq:operator}. Let the sheet $Z_E\subset 
	\wt Z$ be fixed, let $\cR_E(\eps)$ be as in 
	Definition~\ref{def:resonances} and	the associated
	operator-valued function $\sfJ_E(\lm)$ 
	be as in~\eqref{eq:Jacobi2}.
	Then the following equivalence holds
	\begin{equation}\label{eq:res_cond}
		\lm \in \cR_E(\eps)
		\qquad\Longleftrightarrow\qquad 
		\ker(\sfI + \eps\sfJ_E(\lm)) \ne \{0\}.
	\end{equation}
\end{thm}
For $E = \varnothing$ the claim of Theorem~\ref{thm:BS_resonance} follows from Corollary~\ref{cor:BS}\,(ii).
The proof of the remaining part of Theorem~\ref{thm:BS_resonance} is postponed until
Appendix~\ref{app:A}. The argument essentially relies on Krein-type resolvent formula~\cite{S06} for $\Op_\eps$ and on the 
analytic Fredholm theorem~\cite[Thm. 3.4.2]{S2A}.
\begin{remark}\label{rem:symmetric}
	Thanks to compactness of the difference $\sfJ_E(\lm) - \sfJ_0$
	we get by~\cite[Lem. XIII.4.3]{RS-IV} 
	that $\sess(\eps\sfJ_E(\lm)) = \sess(\eps\sfJ_0)) = [-\eps,\eps]$.
	Therefore, the equivalence~\eqref{eq:res_cond}
	can be rewritten as
	\[
		\lm \in \cR_E(\eps)
		\qquad\Longleftrightarrow\qquad 
		-1\in\sd(\eps\sfJ_E(\lm)).
	\]
	Identity $\sfJ_E(\lm)^* = \sfJ_E(\ov\lm)$ 
	combined with~\cite[Rem. III.6.23]{Kato} 
	and with Theorem~\ref{thm:BS_resonance} 
	yields that the set	$\cR_E(\eps)$ is symmetric
	with respect to the real axis.
\end{remark}


%
\section{Localization of bound states and resonances}\label{sec:free}
In this section we prove Theorem~\ref{thm:discr}\,(i) and
Theorem~\ref{thm:res}\,(i). The idea of the proof
is to estimate the norm of $\sfJ_E(\lm)$ and to apply
Corollary~\ref{cor:BS}\,(ii) and Theorem~\ref{thm:BS_resonance}.
\begin{proof}[Proof of Theorem~\ref{thm:discr}\,(i) and Theorem~\ref{thm:res}\,(i)]
	The square of the norm of the operator $\sfJ_E(\lm)$ in~\eqref{eq:Jacobi2} 	
	can be estimated from above by
	\begin{equation}\label{eq:estimate}
	\begin{split}
		\lbra\|\sfJ_E(\lm)\rbra\|^2&  \le 
		\sup_{\xi\in\ell^2(\dN_0), \|\xi\| = 1}
		\|\sfJ_E(\lm)\xi\|^2\\
		& \le
		\sup_{\xi\in\ell^2(\dN_0), \|\xi\| = 1}
		\bigg(
		\sum_{n\in \dN_0}|b_{n}^E(\lm) \xi_{n-1} + b_{n+1}^E(\lm)\xi_{n+1}|^2\bigg)\\
		& \le 
		\sup_{\xi\in\ell^2(\dN_0), \|\xi\| = 1}
		\bigg(
		2\sum_{n\in \dN_0}
		\big(|b_{n}^E(\lm)|^2|\xi_{n-1}|^2 + 
		|b_{n+1}^E(\lm)|^2|\xi_{n+1}|^2\big)\bigg)
		\\
		&\le
		4\sup_{n\in\dN_0}
		|b_n^E(\lm)|^2
		\sup_{\xi\in\ell^2(\dN_0), \|\xi\| = 1}
		\|\xi\|^2 = 4\sup_{n\in\dN}|b_n^E(\lm)|^2,
	\end{split}
	\end{equation}
	where $b_n^E(\lm)$, $n\in\dN_0$, are defined as in~\eqref{eq:bnE}.

	If $\|\eps\sfJ_E(\lm)\| < 1$ holds for a point
	$\lm\in \dC_-$
	then the condition $\ker(\sfI + \eps\sfJ_E(\lm)) \ne \{0\}$ is not satisfied.
	Thus, $\lm$ cannot by Theorem~\ref{thm:BS_resonance}
	be a resonance of $\Op_\eps$ lying on $Z_E^-$
	in the sense of Definition~\ref{def:resonances}.
	In view of estimate~\eqref{eq:estimate} 
	and of~\eqref{eq:bnE}
	to fulfil $\|\eps\sfJ_E(\lm)\| < 1$ it suffices to satisfy
	\[
	 	\frac{n}{|\nu_{n-1} - \lm|^{1/2} |\nu_n - \lm|^{1/2}} < \frac{1}{\eps^2},
	 	\qquad \forall~n\in\dN,
	\]
	or, equivalently, 
	\[
		|\nu_n - \lm| \cdot |\nu_{n-1} - \lm| > \eps^4 n^2,\qquad \forall~n\in\dN.
	\]
	Thus, the claim of Theorem~\ref{thm:res}\,(i) is proven. 
	
	If $\|\eps\sfJ_\varnothing(\lm)\| < 1$ holds for a point
	$\lm\in (0,1/2)$
	then the condition $\ker(\sfI + \eps\sfJ_\varnothing(\lm)) \ne \{0\}$ 
	is not satisfied.
	Thus, by Corollary~\ref{cor:BS}\,(ii),
	$\lm$ is not an eigenvalue of $\Op_\eps$. In view 
	of~\eqref{eq:estimate} and~\eqref{eq:bnE} to fulfil 
	$\|\eps\sfJ_\varnothing(\lm)\| < 1$ it suffices
	to satisfy 
	\begin{equation}\label{eq:quadratic1}
		\big(\nu_{n-1} - \lm\big)\big(\nu_n - \lm\big)
		 = \lm^2 - 2n\lm + n^2-1/4 >
		n^2\eps^4,\qquad \forall~n\in\dN.
	\end{equation}
	The roots of the equation
	$\lm^2 - 2n \lm +  n^2 - 1/4 - n^2\eps^4 = 0$ are given by
	$\lm^\pm_n(\eps) = n \pm \sqrt{1/4 + n^2\eps^4}$.
	Since $\lm^+_n(\eps) > 1/2$ for all $n\in\dN$, 
	the condition~\eqref{eq:quadratic1} yields 
	$\lm_1(\Op_\eps) \ge \min_{n\in\dN} \lm^-_n(\eps)$.
	For $n\in\dN$ we have
	\[
		\lm^-_{n+1}(\eps) - \lm^-_n(\eps)
		 = 
		1 - \frac{(2n+1)\eps^4}{\big(\frac14+ n^2\eps^4\big)^{1/2} +
		\big(\frac14 + (n+1)^2\eps^4\big)^{1/2}}
		\ge
		1 - \frac{(2n+1)\eps^4}{(2n+1)\eps^2}
		= 1 - \eps^2 > 0.
	\]
	Hence, $\min_{n\in\dN} \lm^-_n(\eps) = \lm^-_1(\eps)$
	and the claim of Theorem~\ref{thm:discr}\,(i) follows.
\end{proof}
%

\section{The weak coupling regime: $\eps\arr 0+$}
\label{sec:weak}
In this section we prove Theorem~\ref{thm:discr}\,(ii) and Theorem~\ref{thm:res}\,(ii)-(iii). Intermediate results of this section given in
Lemmata~\ref{lem:aux1} and~\ref{lem:aux2}
are of an independent interest.

First, we introduce some auxiliary operators and functions.
Let $n\in\dN_0$ and the sheet $Z_E \subset \wt Z$ be fixed. 
We make use of notation $\sfP_{kl} := \sfe_{n+k-2}(\cdot,\sfe_{n+l-2})$ with $k,l\in\{1,2,3\}$. Note that  for $n= 0 $ we have $\sfP_{k1} = \sfP_{1k} = 0$ for $k=1,2,3$.
It will also be convenient to decompose the Jacobi matrix $\sfJ_E(\lm)$ in~\eqref{eq:Jacobi2} as
\begin{equation}\label{eq:decomposition0}
	\sfJ_E(\lm)  = \sfS_{n,E}(\lm) + \sfT_{n,E}(\lm), 
\end{equation}	
where the operator-valued functions $\lm\mapsto \sfT_{n,E}(\lm), \sfS_{n,E}(\lm)$ are defined by
\begin{equation}\label{eq:ST_E}
	\sfT_{n,E}(\lm) 
	:= 
	b_{n+1}^E(\lm)\lbra[\sfP_{23} + \sfP_{32} \rbra]
	+
	b_n^E(\lm)\lbra[\sfP_{21} + \sfP_{12}\rbra],\qquad
	\sfS_{n,E}(\lm) := \sfJ_E(\lm)  - \sfT_{n,E}(\lm).
\end{equation}
Clearly, the operator-valued function $\sfS_{n,E}(\cdot)$ is uniformly bounded
on $\dD_{1/2}(\nu_n)$. Moreover, for sufficiently small $r = r(n) \in (0,1/2)$ 
the bounded operator $\sfI + \eps \sfS_{n,E}(\lm)$ is 
at the same time boundedly invertible for all $(\eps,\lm) \in \Omega_r(n) := \dD_r\times \dD_r(\nu_n)$.
Thus, the operator-valued function
\begin{equation}\label{eq:R2}
	\sfR_{n,E}(\eps,\lm) := 
	\big(\sfI + \eps\sfS_{n,E}(\lm)\big)^{-1},
\end{equation}
is well-defined and analytic on $\Omega_r(n)$ and, in particular, $\sfR_{n,E}(0,\nu_n) = \sfI$.
Furthermore, we introduce auxiliary scalar functions
$\Omega_r(n)\ni (\eps,\lm)\mapsto f_{kl}^E(\eps,\lm)$ by 
\begin{equation}\label{eq:Fnm_E}
	f_{kl}^E(\eps,\lm) 
	:= \big(\sfR_{n,E}(\eps,\lm)\sfe_{n+k-2},\sfe_{n+l-2}\big),
	\qquad k,l\in\{1,2,3\}.
\end{equation} 
Thanks to $\sfR_{n,E}(0,\nu_n) = \sfI$ we have $f_{kl}^E(0,\nu_n) = \delta_{kl}$.
Finally, we introduce $3\times 3$ matrix-valued function
\begin{equation}\label{eq:Amatrix}
	\dD_r\times \dD_r^\times(\nu_n)\ni (\eps,\lm)\mapsto\sfA_{n,E}(\eps,\lm) :=  \big(a_{kl}^E(\eps,\lm)\big)_{k,l = 1}^{3,3}
\end{equation}  
with the entries given for $k,l=1,2,3$ by
\begin{equation}\label{eq:akl}
	a_{kl}^E(\eps,\lm) 
	:=
	b_n^E(\lm)\big(f_{1k}^E(\eps,\lm)\delta_{2l} + f_{2k}^E(\eps,\lm)\delta_{1l}\big)
	+
	b_{n+1}^E(\lm)
	\big(f_{2k}^E(\eps,\lm)\delta_{3l} + f_{3k}^E(\eps,\lm)\delta_{2l}\big).
\end{equation}
We remark that $\rank\sfA_{n,E}(\eps,\lm) \le 2$ due to linear
dependence between the first and the third columns in $\sfA_{n,E}(\eps,\lm)$.

In the first lemma we derive an implicit scalar equation which characterises those points $\lm \in\dC\setminus [\nu_0,+\infty)$ near $\nu_n$ for which the condition $\ker(\sfI+ \eps\sfJ_E(\lm))\ne \{0\}$
is satisfied under additional assumption that $\eps > 0$ is small enough. This equation can be used
to characterise the `true' resonances for $\Op_\eps$ as well as the weakly coupled bound state if $n = 0$ and $E = \varnothing$.
\begin{lem}\label{lem:aux1}
	Let the self-adjoint operator $\Op_\eps$, $\eps \in (0,1)$, 
	be as in~\eqref{eq:operator}.
	Let $n\in\dN_0$ and the sheet $Z_E\subset \wt Z$ be fixed.
	Let $r = r(n) > 0$ be chosen as above.
	Then for all $\eps \in (0,r)$ 
	a point $\lm \in \dD_r(\nu_n) \setminus [\nu_0,\infty)$ is a resonance of $\Op_\eps$ 
	on $Z_E$ \iff
	\[
		\det\big(\sfI + \eps\sfA_{n,E}(\eps,\lm)\big) = 0.
	\]
\end{lem}

\begin{proof}
	Using the decomposition~\eqref{eq:decomposition0} of 
	$\sfJ_E(\lm)$ and the auxiliary operator in~\eqref{eq:R2} we find
	\begin{equation}\label{eq:reduction}
	\begin{split}
		\dim\ker \lbra(\sfI +\eps\sfJ_E(\lm)\rbra) 
		& = 
		\dim\ker\lbra(\sfI + 
		\eps\sfS_{n,E}(\lm) + \eps\sfT_{n,E}(\lm)\rbra)\\
		& = 
		\dim\ker\lbra(\sfI + \eps\sfR_{n,E}(\eps,\lm)\sfT_{n,E}(\lm)\rbra).
	\end{split}				   	
	\end{equation}
	%
	Note that 
	\[
		\rank(\sfR_{n,E}(\eps,\lm)\sfT_{n,E}(\lm)) \le \rank(\sfT_{n,E}(\lm)) \le 3 
	\]
	and, hence, using~\cite[Thm. 3.5\,(b)]{S05}, we get
	\begin{equation}\label{eq:kernel_det}
   		\dim\ker\lbra(\sfI +  \eps \sfR_{n,E}(\eps,\lm)\sfT_{n,E}(\lm)\rbra) \ge 1 
		\qquad \Longleftrightarrow\qquad 
		\det\lbra(\sfI + \eps \sfR_{n,E}(\eps,\lm)\sfT_{n,E}(\lm)\rbra) = 0.
	\end{equation}
	For the orthogonal projector 
	$\sfP := \sfP_{11} + \sfP_{22} + \sfP_{33}$ the identity
	$\sfT_{n,E}(\lm) = \sfT_{n,E}(\lm)\sfP$ is straightforward.
	Hence, employing~\cite[IV.1.5]{GK69} we find
	\begin{equation}\label{eq:determinant2}
	\begin{split}
		\det\lbra(\sfI + \eps \sfR_{n,E}(\eps,\lm)\sfT_{n,E}(\lm)\rbra) 
		& = 
		\det\lbra(\sfI + \eps \sfR_{n,E}(\eps,\lm)\sfT_{n,E}(\lm)\sfP\rbra)\\
		&=
		\det\lbra(\sfI + \eps \sfP\sfR_{n,E}(\eps,\lm)\sfT_{n,E}(\lm)\rbra).
	\end{split}	
	\end{equation}
	For $k,l\in\{1,2,3\}$ we can write the following identities
	\[
	\begin{split}
		\sfP_{kk}\sfP\sfR_{n,E}(\eps,\lm)\sfT_{n,E}(\lm)\sfP_{ll} 
		& =
		\sfP_{kk}\sfR_{n,E}(\eps,\lm)
		\Big(
		b_{n}^E(\lm)\lbra[\sfP_{21} + \sfP_{12}\rbra]+
		b_{n+1}^E(\lm)	\lbra[\sfP_{23} + \sfP_{32}  \rbra] \Big)\sfP_{ll}
		\\
		&=
		\sfP_{kk}\sfR_{n,E}(\eps,\lm)
		\Big(
		b_{n}^E(\lm)\lbra[\sfP_{2l}\delta_{1l} + \sfP_{1l}\delta_{2l}\rbra]
		+
		b_{n+1}^E(\lm)	\lbra[\sfP_{2l}\delta_{3l} + \sfP_{3l}\delta_{2l}  \rbra]
		\Big) \\
		& = 
		\sfP_{kl}b_{n}^E(\lm)\lbra[f_{2k}^E(\eps,\lm)\delta_{1l} + 
		f_{1k}^E(\eps,\lm)\delta_{2l}\rbra]\\
 		&\qquad\qquad\qquad 
 		+ \sfP_{kl} b_{n+1}^E(\lm)\lbra[f_{2k}^E(\eps,\lm)\delta_{3l} + f_{3k}^E(\eps,\lm)\delta_{2l}  \rbra]\\
		&= a_{kl}^E(\eps,\lm) \sfP_{kl}
	\end{split}	
	\]
	with $f_{kl}^E$ as in~\eqref{eq:Fnm_E}, and as a result
	we get
	\[
	\begin{split}
		\sfP\sfR_{n,E}(\eps,\lm)\sfT_{n,E}(\lm) =
		\sum_{k = 1}^3\sum_{l = 1}^3 a_{kl}^E(\eps,\lm) \sfP_{kl},
	\end{split}	
	\]
	with $a_{kl}^E(\eps,\lm)$ as in~\eqref{eq:akl}.
	Hence, the determinant in~\eqref{eq:determinant2} can be expressed as	
	\[
		\det\lbra(\sfI + \eps \sfR_{n,E}(\eps,\lm)\sfT_{n,E}(\lm)\rbra) 
		= 
		\det(\sfI + \eps \sfA_{n,E}(\eps,\lm))
	\]
	where on the right-hand side we have the
	determinant of the $3\times 3$ matrix 
	$\sfI + \eps \sfA_{n,E}(\eps,\lm)$; \cf~\eqref{eq:Amatrix}.
	The claim of lemma follows then from~\eqref{eq:reduction},
	~\eqref{eq:kernel_det}, and  Theorem~\ref{thm:BS_resonance}.
\end{proof}
In the second lemma we establish the existence and investigate properties
of solutions of the scalar equation in Lemma~\ref{lem:aux1}.
To this aim it is natural to try to apply the \emph{analytic implicit function theorem}.
The main obstacle that makes a direct application
of the implicit function theorem difficult lies
in the fact that $\lm\mapsto\det(\sfI + \eps\sfA_{n,E}(\eps,\lm))$
is not analytic near $\nu_n$ due to the cut on the real axis. 
We circumvent this obstacle by applying the analytic implicit function
theorem to an auxiliary function which is analytic in the disc and 
has values in different sectors of this disc that are in direct
correspondence with the values of $\lm\mapsto\det(\sfI + \eps\sfA_{n,E}(\eps,\lm))$ 
on the four different sheets in $\wt Z$ which are mutually adjacent
in a proper way.
\begin{assumption}\label{assump}
	Let $n\in\dN_0$ and the sheet $Z_E\subset \wt Z$ be fixed.
	Let the sheets $Z_F$, 
	$Z_G$ and $Z_H$ be such that
	$Z_E \sim_{n-1} Z_F$, $Z_F \sim_n Z_G$ and 
	$Z_G \sim_{n-1} Z_H$.
	For $r > 0$ let the matrix-valued function 
	$\dD_r\times\dD_r^\times \ni 
	(\eps,\kp)\mapsto \sfB_{n,E}(\eps,\kp)$ 
	be defined by
	\[
			\sfB_{n,E}(\eps,\kp) := 
			\begin{cases}
				\sfA_{n,E}(\eps, \nu_n - \kp^4),&
				\quad\arg\kp \in\Phi_E  :=
				(-\pi,-\tfrac{3\pi}{4}]\cup (0,\tfrac{\pi}{4}] ,\\
				\sfA_{n,F}(\eps, \nu_n - \kp^4),& 
				\quad\arg\kp \in \Phi_F := 
				(-\tfrac{3\pi}{4},-\tfrac{\pi}{2}]\cup (\tfrac{\pi}{4},
	\tfrac{\pi}{2}],\\
				\sfA_{n,G}(\eps, \nu_n - \kp^4),&
				 \quad\arg\kp \in \Phi_G :=
				 (-\tfrac{\pi}{2},-\tfrac{\pi}{4}]
				 \cup (\tfrac{\pi}{2},\tfrac{3\pi}{4}],
				\\
				\sfA_{n,H}(\eps, \nu_n - \kp^4),
				&
				\quad\arg\kp \in\Phi_H := 
				(-\tfrac{\pi}{4},0]\cup (\tfrac{3\pi}{4},\pi].
			\end{cases}
	\]
\end{assumption}
Tracing the changes in the characteristic vector
along the path $Z_E \sim_{n-1} Z_F \sim_n Z_G \sim_{n-1} Z_H$ one
easily verifies that $Z_H \sim_n Z_E$.
Thanks to that $\sfB_{n,E}$ is analytic on 
$\dD_r\times\dD_r^\times$ for sufficiently small $r > 0$ which is essentially a consequence
of componentwise analyticity in $\dD_r$ of vector-valued function
\[
	\kp \mapsto R_\bullet(\nu_n-\kp^4),\qquad
	\bullet \in \{E,F,G,H\}\quad\text{for}\quad\arg\kp \in \Phi_\bullet,
\]
where $R_\bullet$ is as in~\eqref{eq:RE}.
\begin{lem}\label{lem:aux2}
	Let $n\in\dN_0$ and the sheet $Z_E\subset \wt Z$ be fixed.
	Set $(\frp,\frq,\frr) := (l_{n-1}^E,l_n^E,l_{n+1}^E)$.
	Let the matrix-valued function $\sfB_{n,E}$ 
	be as in Assumption~\ref{assump}. 
	Then the implicit scalar equation
	\[
		\det\big(\sfI + \eps \sfB_{n,E}(\eps,\kp)\big) = 0
	\]	
	has exactly two solutions $\kp_{n,E,j}(\cdot)$
	analytic near $\eps = 0$ such that $\kp_{n,E,j}(0) = 0$,
	satisfying $\det(\sfI + \eps 
	\sfB_{n,E}(\eps,\kp_{n,E,j}(\eps))) = 0$
	pointwise for sufficiently small $\eps > 0$, and having asymptotic expansions 
	\begin{equation}\label{eq:asymp_kp}
		\kp_{n,E,j}(\eps) 
		=
		\eps\frac{(z_{n,E})^{1/2}_j}{2}
		+ \cO(\eps^2), \qquad \eps\arr 0+,
	\end{equation}
	where 
	$z_{n,E} = (-1)^{\frq + \frr}(n+1) + (-1)^{\frp + \frq+1}n\ii$.
\end{lem}
\begin{proof}
	First, we introduce the shorthand notations 
	\[
		u(\kp) := b_n^\bullet(\nu_n -\kp^4),
		\quad
		v(\kp) := b_{n+1}^\bullet(\nu_n -\kp^4),\qquad
		\bullet \in \{E,F,G,H\}\quad \text{for}\quad
		\arg\kp\in \Phi_\bullet.
	\]
	Let $b_{kl}$ with $k,l\in\{1,2,3\}$ 
	be the entries of the matrix-valued
	function $\sfB_{n,E}$. Furthermore, define the scalar functions 
	$X = X(\eps,\kp)$, $Y = Y(\eps,\kp)$, and $Z = Z(\eps,\kp)$ by
	\begin{equation}\label{eq:XYZ}
	\begin{split}
		X & := b_{11} + b_{22}+ b_{33},\\
		Y & :=
		b_{11} b_{22} + b_{22} b_{33}+ b_{11} b_{33}
        			- b_{13} b_{31} - b_{12}b_{21}- b_{23} b_{32},\\
		Z & := b_{11} b_{22} b_{33}+ 
		b_{13} b_{32} b_{21}
				+ b_{12}b_{23}b_{31}
		 - b_{13} b_{31} b_{22}
				-b_{12} b_{21} b_{33}- b_{11} b_{23}b_{32}.
	\end{split}
	\end{equation}
	Employing an elementary formula 
	for the determinant of $3\times 3$ matrix,
	the equation $\det(\sfI+ \eps\sfB_{n,E}(\eps,\kp)) = 0$
	can be equivalently written as
	\begin{equation}\label{eq:quadratic0}
		1 + \eps X(\eps,\kp) + \eps^2 Y(\eps,\kp) + \eps^3 Z(\eps,\kp) = 0.
	\end{equation}
	By purely algebraic argument one can derive 
	from~\eqref{eq:akl} that $Z  = 0$.
	%
	Hence,~\eqref{eq:quadratic0} simplifies to
	$1 + \eps X(\eps,\kp) + \eps^2 Y(\eps,\kp) = 0$.	
	Introducing new parameter $t := \eps/\kp$ we can further rewrite
	this equation as
	\begin{equation}\label{eq:quadratic}
		1 + t \kp X(\eps,\kp) + t^2 \kp^2 Y(\eps,\kp) = 0.
	\end{equation}	
	Note also that the coefficients
	$(\eps,\kp)\mapsto \kp X(\eps,\kp), \kp^2 Y(\eps,\kp)$ 
	of the quadratic equation~\eqref{eq:quadratic}
	are analytic in $\dD_r^2$.
	For each fixed pair $(\eps,\kp)$ the equation~\eqref{eq:quadratic} 
	has (in general) two distinct roots $t_j(\eps,\kp)$, $j=0,1$.
	The condition $\det(\sfI + \eps\sfB_{n,E}(\eps,\kp)) = 0$ 
	with $\kp \ne 0$ holds \iff at least one of the two conditions
	\begin{equation}\label{eq:F_j}
		f_j(\eps,\kp) : = \eps - \kp t_j(\eps,\kp) = 0,\qquad j=0,1,
	\end{equation}
	is satisfied. Using analyticity of $u(\cdot)$ and $v(\cdot)$ near $\kp = 0$, 
	we compute 
	\[
	\begin{split}
		\lim_{\kp\arr 0} 
		\kp	u 
		& =	\lim_{r\arr 0+} re^{\ii\pi/8}u(re^{\ii\pi/8})\\
		& =
		\lim_{r\arr 0+} \frac{n^{1/2}}{2}
		\frac{re^{\ii\pi/8}}{((-1 + \ii r^4)^{1/2}_\frp
		(\ii r^4)^{1/2}_\frq)^{1/2}} 		 =
		\frac{n^{1/2}e^{\ii\pi/8}}{2((-1)^{\frp+\frq} 
		\ii	e^{\ii\pi/4})^{1/2}}, \\
		\lim_{\kp\arr 0} 
		\kp	v 
		& = \lim_{r\arr 0+} re^{\ii\pi/8}v(re^{\ii\pi/8})\\
		& =
		\lim_{r\arr 0+} 
		\frac{(n+1)^{1/2}}{2}
		\frac{re^{\ii\pi/8}}{((\ii r^4)^{1/2}_\frq(1+\ii r^4)^{1/2}_\frr)^{1/2}} 
		= 
		\frac{(n+1)^{1/2}e^{\ii\pi/8}}{2((-1)^{\frq+\frr} e^{\ii\pi/4})^{1/2}}.
	\end{split}	
	\]
	Hence, we get
	\begin{equation*}\label{eq:limits}
	\begin{split}
		\lim_{\eps,r\arr 0+} 
		re^{\ii\pi/8}
		b_{kl}(\eps,re^{\ii\pi/8}) 
		& =
		\lim_{r\arr 0+}re^{\ii\pi/8} u(re^{\ii\pi/8})
		\big(f_{2k}^E(\bfo)\delta_{3l} + f_{3k}^E(\bfo)\delta_{2l}\big)\\
		&\qquad\qquad\qquad +
		\lim_{r\arr 0+}re^{\ii\pi/8} v(re^{\ii\pi/8})
		\big(f_{1k}^E(\bfo)\delta_{2l} + f_{2k}^E(\bfo)\delta_{1l}\big)\\ 		
		&=
		\frac{n^{1/2}e^{\ii\pi/8}\big(\delta_{2k}\delta_{3l} +
		\delta_{3k}\delta_{2l}\big) }{2((-1)^{\frp+\frq} 
		\ii	e^{\ii\pi/4})^{1/2}}
		+
		\frac{(n+1)^{1/2}e^{\ii\pi/8}
		\big(\delta_{1k}\delta_{2l} + \delta_{2k}\delta_{1l}\big)}
		{2((-1)^{\frq+\frr} e^{\ii\pi/4})^{1/2}}.
		\end{split}	
	\end{equation*}
	Combining this with~\eqref{eq:XYZ} we end up with
	\begin{equation*}\label{eq:limits1.5}
	\begin{split}
		\lim_{(\eps,\kp)\arr \bfo} \kp X &  = 
		\lim_{\eps,r\arr 0+} re^{\ii\pi/8}X(\eps,re^{\ii\pi/8})
		=
		\lim_{\eps,r\arr 0+} re^{\ii\pi/8}
		\big[b_{11} +  b_{22} + b_{33}\big](\eps,re^{\ii\pi/8}) = 0,\\
		\lim_{(\eps,\kp)\arr \bfo} \kp^2 Y
		 & = \lim_{\eps,r\arr 0+} r^2e^{\ii\pi/4}Y(\eps,re^{\ii\pi/8})\\
		& = \lim_{\eps,r\arr 0+} r^2e^{\ii\pi/4}
		\big[b_{11} b_{22} + b_{22} b_{33}+ b_{11} b_{33}
        			- b_{13} b_{31} - b_{12}b_{21}- b_{23} b_{32}\big](\eps,
        			re^{\ii\pi/8})\\
		 & = \lim_{\eps,r\arr 0+} r^2e^{\ii\pi/4}
		\big[-b_{12}b_{21} - b_{23} b_{32}\big](\eps,re^{\ii\pi/8})\\
		& =
		 -
		\lbra(
		\frac{n^{1/2}e^{\ii\pi/8}}{2((-1)^{\frp+\frq} 
		\ii	e^{\ii\pi/4})^{1/2}}
		\rbra)^2
		-
		\lbra(
		\frac{(n+1)^{1/2}e^{\ii\pi/8}}{2((-1)^{\frq+\frr} e^{\ii\pi/4})^{1/2}}
		\rbra)^2   \\
		& =  
		 -\frac{(-1)^{\frp + \frq + 1}n\ii}{4} 
		 -
		 \frac{(-1)^{\frq + \frr}(n+1)}{4} 
		 =  
		 -\frac{z_{n,E}}{4}.
  	\end{split}	  
	\end{equation*}
	Hence, the roots $t_j(\eps,\kp)$ of~\eqref{eq:quadratic}
	converge in the limit $(\eps,\kp)\arr \bfo$ to the roots
	$2\big[(z_{n,E})^{1/2}_j\big]^{-1}$, $j=0,1$, of 
	the quadratic equation $z_{n,E}t^2 - 4 = 0$.
	Moreover, analyticity of the coefficients in equation~\eqref{eq:quadratic},
	the above limits, and the formula for the roots of a quadratic equation
	imply analyticity of the functions $(\eps,\kp)\mapsto t_j(\eps,\kp)$
	near $\bfo$.
	\vspace{0.7mm}

	\noindent\underline{\emph{Step 2.}}
	The partial derivatives of $f_j$ in~\eqref{eq:F_j}
	with respect to $\eps$ and $\kp$
	are given by
	$\p_\eps f_j  =  1 - \kp \p_\eps t_j$
	and $\p_\kp f_j  =  - t_j - \kp\p_{\kp} t_j$.
	Analyticity of $t_j$ near $\bfo$
	implies $(\p_\eps f_j)(\bfo)  = 1$ and 
	$(\p_\kp f_j)(\bfo) = - t_j$.
	In particular, we have shown that $(\p_\kp f_j)(\bfo) \ne 0$.
	Since the functions $f_j(\cdot)$ are analytic near $\bfo$
	and satisfy $f_j(\bfo) = 0$, we can apply the
	analytic implicit function theorem~\cite[Thm. 3.4.2]{S2A} 
	which yields existence of a unique function $\kp_j(\cdot)$, 
	analytic near $\eps = 0$ such that $\kp_j(0) = 0$
	and that $f_j(\eps, \kp_j(\eps))  = 0$ holds pointwise.
	Moreover, the derivative of $\kp_j$ at $\eps = 0$ can be expressed as
	\begin{equation}\label{eq:kp_der}
		\kp^\pp_j(0)
		= - 
		\frac{(\p_\eps f_j)(\bfo)}{(\p_\kp f_j)(\bfo)} 
		= 
		\frac{1}{t_j(\bfo)}.
	\end{equation}
	Hence, we obtain Taylor expansion for $\kp_j$ near $\eps = 0$
	\begin{equation*}\label{eq:Taylor}
		\kp_j(\eps)  = 
		\kp_j(0) 
		+ 
		\kp^\pp_j(0)\eps + \cO(\eps^2) 
		= 
		\frac{\eps}{t_j(\bfo)} + \cO(\eps^2)
		= 
		\eps\frac{(z_{n,E})^{1/2}_j}{2} + \cO(\eps^2)\qquad \eps\arr 0+.
	\end{equation*}
	The functions $\kp_j$, $j=0,1$, satisfy all the requirements
	in the claim of the lemma.
\end{proof}
Now we are prepared to prove Theorem~\ref{thm:discr}\,(ii) and 
Theorem~\ref{thm:res}\,(ii)-(iii) from the introduction.
\begin{proof}[Proof of Theorem~\ref{thm:discr}\,(ii)]
	By Proposition~\ref{prop:known}\,(iv) we have
	$\cN_{1/2}(\Op_\eps) = 1$ for all sufficiently small
	$\eps > 0$. Recall that we denote by $\lm_1(\Op_\eps)$ 
	the corresponding unique eigenvalue. Thus, we have by Lemma~\ref{lem:aux1}
	\[
		\det(\sfI+ \eps\sfA_{0,\varnothing}(\eps,\lm_1(\Op_\eps))) 
		= 0.
	\]
	Using the construction of Assumption~\ref{assump}
	for the physical sheet and $n = 0$, we obtain
	\[
		\det(\sfI + \eps\sfB_{0,\varnothing}(
		\eps, (\nu_0 - \lm_1(\Op_\eps))^{1/4})) 
		=
		\det(\sfI+ \eps\sfA_{0,\varnothing}(\eps,\lm_1(\Op_\eps))) = 0,
	\]
	where we have chosen the principal branch for $(\cdot)^{1/4}$.
	Thus, by Lemma~\ref{lem:aux2} we get
	\[
		(\nu_0 - \lm_1(\Op_\eps))^{1/4} = \frac{\eps}{2}
		+ \cO(\eps^2),\qquad \eps \arr 0+,
	\]
	where we have used the fact that $z_{0,\varnothing} = 1$.
	Hence, taking the fourth power of the left
	and right hand sides in the above equation we arrive at
	\[
		\lm_1(\Op_\eps) = \nu_0 - \frac{\eps^4}{16}
		+ \cO(\eps^5),\qquad \eps \arr 0+.\qedhere
	\]
\end{proof}

\begin{proof}[Proof of Theorem~\ref{thm:res}\,(ii)-(iii)]
	Let $n\in\dN$ and the sheet $Z_E \subset \wt Z$ be fixed.
	Let us repeat the construction of Assumption~\ref{assump}. 
	By Lemma~\ref{lem:aux2} we infer that
	there exist exactly
	two analytic solutions $\kp_{n,E,j}$, $j=0,1$ of the 
	implicit scalar equation 
	$\det(\sfI+ \eps\sfB_{n,E}(\eps,\kp)) = 0$ 
	such that $\kp_{n,E,j}(0) = 0$. 
	It can be checked that
	both solutions correspond to the same resonance
	and it suffices to analyse the solution 
	$\kp_{n,E} := \kp_{n,E,0}$ only.

	For all small enough $\eps > 0$ 
	the asymptotics~\eqref{eq:asymp_kp} yields
	\[
		\arg(\kp_{n,E}(\eps)) = 
		\frac{1}{2}\arg(z_{n,E}) \in \Phi_E,
		\quad \text{\iff} \quad n\in\cS(E).
	\]
	Hence, if $n\in\dN \setminus \cS(E)$,
	Lemmata~\ref{lem:aux1} and~\ref{lem:aux2}
	imply that there will be no resonances in a neighbourhood of
	the point $\lm = \nu_n$ lying on $Z_E^-$ for 
	sufficiently small $\eps > 0$.
	Thus, we have proven Theorem~\ref{thm:res}\,(iii).
	While if $n\in \cS(E)$ we get by Lemmata~\ref{lem:aux1}
	and~\ref{lem:aux2}
	that there will be exactly one resonance 
	\[
		\lm_n^E(\Op_\eps) =  \nu_n - (\kp_{n,E}(\eps))^4,
	\]
	in a neighbourhood of the point $\lm = \nu_n$ 
	lying on $Z_E^-$ for sufficiently small $\eps > 0$ and its asymptotic
	expansion is a direct consequence of the asymptotic
	expansion~\eqref{eq:asymp_kp} given in Lemma~\ref{lem:aux2}.	
	Thus, the claim of Theorem~\ref{thm:res}\,(ii) follows.
\end{proof}


%



\begin{appendix}
\section{Krein's formula, meromorphic continuation of resolvent, and
condition on resonances}\label{app:A}
In this appendix we use Krein's resolvent formula for Smilansky Hamiltonian
to prove Proposition~\ref{prop:cont} and Theorem~\ref{thm:BS_resonance}
on meromorphic continuation of $(\Op_\eps - \lm)^{-1}$ to $\wt Z$.
The proposed continuation procedure is of an iterative nature wherein we, first, extend
$(\Op_\eps - \lm)^{-1}$ to the sheets adjacent to the physical sheet, 
then to the sheets which are adjacent to the sheets being adjacent to the physical sheet and so on. 

To this aim we define for $n\in\N_0$ the scalar functions
$\dC \setminus [\nu_0,+\infty) \mapsto 	y_n(\lm)$ and
$(\dC \setminus [\nu_0,+\infty))\times \R \mapsto \eta_n(\lm;x)$ by
\begin{equation}\label{eq:functions}
	y_n(\lm)  := r_n(\lm)\sqrt{\nu_n},\qquad
	\eta_n(\lm;x)  := \nu_n^{1/4}\exp(-r_n(\lm)|x|),
\end{equation}	
where $r_n(\cdot)$, $n\in\N_0$, is as in~\eqref{eq:r_n}.
Next, we introduce the following operator-valued function 
\[
	\sfT(\lm)\colon \ell^2(\dN_0)\arr \cH, 
	\qquad 
	\sfT(\lm)\{c_n\} := \{c_n \eta_n(\lm;x)\}.
\]
For each fixed $\lm \in \C\setminus [\nu_0,+\infty)$
the operator $\sfT(\lm)$ is bounded and everywhere defined
and the adjoint of $\sfT(\ov\lm)$ acts as 
\[
	\sfT(\ov\lm)^*\{u_n\} \sim \{I_n(\lm;u_n)\}_{n\in\dN_0},
	\qquad I_n(\lm;u) := \int_\R \eta_n(\lm;x)u(x)\dd x.
\]
With these preparations, the resolvent difference of 
$\Op_\eps$ and $\Op_0$ can be expressed by~\cite[Thm. 6.1]{S06} 
(see also~\cite[Sec. 6]{NS06}) as follows
\begin{equation}\label{eq:Krein}
	(\Op_\eps - \lm)^{-1} 
	= 
	(\Op_0 - \lm)^{-1} 
	+ 
	\sfT(\lm)\sfY(\lm)
	\big[\big(\sfI + \eps\sfJ_\varnothing(\lm)\big)^{-1} - \sfI\big] 
	\sfY(\lm)\sfT(\ov\lm)^*,\quad \lm \in \C\setminus [\nu_0,+\infty),
\end{equation}
where $\Op_0$ is 
the Smilansky Hamiltonian with $\eps = 0$, 
$\sfY(\lm) := \diag\{(2y_n(\lm))^{-1/2}\}$,
and $\sfJ_\varnothing(\lm)$ is as in~\eqref{eq:Jacobi2}.
The formula~\eqref{eq:Krein} can be viewed as a particular case of abstract Krein's 
formula (see \eg \cite{BL12, BGP08, DM91}) for the resolvent difference of 
two self-adjoint extensions of their common densely defined symmetric restriction.
\begin{proof}[Proof of Proposition~\ref{prop:cont} and Theorem~\ref{thm:BS_resonance}]
	Let us fix $n\in\dN_0$ and a sheet $Z_E\subset\wt Z$. 
	We denote by $\sfR_n(\lm)$ the resolvent of the
	self-adjoint operator $H^2(\dR)\ni f\mapsto -f'' + \nu_n f$ in 
	the Hilbert space $L^2(\dR)$.
	We can express the function $\frr_{n,\eps}^\varnothing(\cdot;u)$ 
	in~\eqref{eq:frr} using Krein's formula~\eqref{eq:Krein} as
	\[
	\begin{split}
		\frr_{n,\eps}^\varnothing(\lm;u)
 		& = 
		\big\langle (\Op_\eps - \lm)^{-1} u\otimes \sfe_n, u\otimes\sfe_n\big\rangle\\
		& = 
		\big\langle (\Op_0 - \lm)^{-1} u\otimes \sfe_n, u\otimes \sfe_n\big\rangle\\ 
		& \qquad \qquad \qquad + 
		\lbra(
			\sfY(\lm)\Big[\lbra(\sfI + \eps\sfJ_\varnothing(\lm)\rbra)^{-1} - 
			\sfI\Big]\sfY(\lm)  
			\sfT(\ov \lm)^*u\otimes \sfe_n, \sfT(\lm)^* u\otimes \sfe_n
		\rbra)\\
		& = 
		(\sfR_n(\lm)u,u)_\dR
		+ 
		I_n(\lm;u)\ov{I_n(\ov\lm;u)}
		\lbra(
			\big[
				\lbra(\sfI + \eps\sfJ_\varnothing(\lm)\rbra)^{-1} - \sfI
			\big]\sfY(\lm)  
			 \sfe_n,\sfY(\lm)^*\sfe_n
		\rbra)		\\
		& =
		(\sfR_n(\lm)u,u)_\dR
		+
		\frac{I_n(\lm;u)I_n(\lm;\ov{u})}{2y_n(\lm)}
		\bigg[
			\lbra(\big(\sfI + \eps\sfJ_\varnothing(\lm)\big)^{-1}\sfe_n,\sfe_n\rbra) 
			- 1
		\bigg].
	\end{split}	
	\]
	%
	Since  $(\sfR_n(\lm)u,u)_\dR$, $y_n(\lm)$, 
	$I_n(\lm;u)$, and $I_n(\lm;\ov{u})$ can be easily analytically
	continued to $\wt Z$, to extend $\frr_{n,\eps}^\varnothing(\cdot; u )$ 
	meromorphically to the other sheets of the component $\wt Z$ it suffices to
	extend
	\[
		\frs_{n,\eps}^\varnothing(\lm) := 
		\lbra(\lbra(\sfI + \eps\sfJ_\varnothing(\lm)\rbra)^{-1}\sfe_n,\sfe_n\rbra),
	\]
	meromorphically from $Z_\varnothing$ to $\wt Z$. 
	The poles of the meromorphic extension of $\frs_{n,\eps}^\varnothing(\cdot)$
	can be identified with the resonances of $\Op_\eps$ in the sense 
	of Definition~\ref{def:resonances}.

	To this aim we set by definition
	\[
		\frs_{n,\eps}^E(\lm) 
		:= 
		\lbra(\lbra(\sfI + \eps\sfJ_E(\lm)\rbra)^{-1}\sfe_n,\sfe_n\rbra),
	\]
	for any $\lm \in \dC\setminus [\nu_0,+\infty)$ 
	such that $-1\notin \s(\eps\sfJ_E(\lm))$.
	In what follows let $Z_E$ and $Z_F$ be two sheets of $\wt Z$
	such that $Z_E\sim_{n-1} Z_F$ with $n\in\dN_0$\footnote{Recall that 
	for any	sheet $Z_E$ holds $Z_E\sim_{-1} Z_E$.}.
	Suppose that $\lm\mapsto \frs_{n,\eps}^E(\cdot)$ 
	is well defined and meromorphic either on $Z_E^+$ or on $Z_E^-$.
	Next, we extend $\lm\mapsto \frs_{n,\eps}^E(\cdot)$ meromorphically
	from $Z_E^\pm$ to $Z_F^\mp$. Without loss of generality we restrict
	our attention to the case that $\lm\mapsto \frs_{n,\eps}^E(\cdot)$ 
	is meromorphic on $Z_E^+$ and extend it meromorphically to $Z_F^-$.
	On the open set $\Omega_n := \dC_+ \cup \dC_- \cup (\nu_{n-1},\nu_n)$ 
	the operator-valued function 
	\[
		\sfJ_{EF}(\lm) :=	
		\begin{cases}
			\sfJ_E(\lm), &\quad \lm \in \dC_+,\\
			\sfJ_F(\lm), &\quad \lm \in \Omega_n\setminus\dC_+,
		\end{cases}
	\]
	is analytic which is essentially a consequence of analyticity 
	on $\Omega_n$ of the entries $b_m^\bullet(\lm)$
	(with $\bullet = E$ for $\lm\in\dC_+$ and $\bullet = F$ for 
	$\lm \in\dC_-$) for the underlying Jacobi matrix.
	Thus, the operator-valued function
	\[
		\Omega_n\ni \lm\mapsto \sfA^{EF}_\eps(\lm) :=
			\eps\lbra(\sfI + \eps\sfJ_0\rbra)^{-1}\lbra(\sfJ_{EF}(\lm) - \sfJ_0\rbra)
	\]
	is also analytic on $\Omega_n$ because of the analyticity of $\sfJ_{EF}(\lm)$. 
	Furthermore, the values of $\sfA^{EF}_\eps(\cdot)$ are compact operators thanks 
	to compactness of the difference $\sfJ_{EF}(\lm) - \sfJ_0$.
	Taking into account that
	\[
		\lbra(\lbra(\sfI + \sfA^{EF}_\eps(\lm)\rbra)^{-1}\sfe_n,
		\lbra(\sfI +\eps\sfJ_0\rbra)^{-1}\sfe_n\rbra)
		= 
		\begin{cases}
				\frs_{n,\eps}^E(\lm), \qquad \lm \in\dC_+,\\
				\frs_{n,\eps}^F(\lm), \qquad \lm \in\Omega_n\setminus\dC_+,
		\end{cases}
	\]
	we obtain from the analytic Fredholm theorem~\cite[Thm. VI.14]{RS-I} that
	$\dC_-\ni \lm \mapsto \frs_{n,\eps}^F(\lm)$ is a meromorphic continuation
	of $\dC_+\ni \lm \mapsto \frs_{n,\eps}^E(\lm)$ across the interval
	$(\nu_{n-1},\nu_n)$ and that the poles of 
	$\dC_- \ni \lm \mapsto \frs_{n,\eps}^F(\lm)$ 
	satisfy the condition
	\[
		\ker\lbra(\sfI + \eps\sfJ_F(\lm)\rbra) \ne \{0\},\qquad \lm\in\dC_-.
	\]
	Starting from the physical sheet $Z_\varnothing$ we use the above
	procedure iteratively to extend $\frs_{n,\eps}^\varnothing(\cdot)$
	meromorphically to the whole of $\wt Z$ thus proving
	Proposition~\ref{prop:cont} and Theorem~\ref{thm:BS_resonance}.
\end{proof}	
%

\section*{Acknowledgements}
This research was supported by the Czech Science Foundation
(GA\v{C}R) within the project 14-06818S.

\end{appendix}

\newcommand{\etalchar}[1]{$^{#1}$}


%
\end{document}